\documentclass{llncs}
\usepackage{amssymb}
\usepackage{todonotes}
\usepackage{lineno}
\usepackage{graphicx}
\usepackage{subcaption}
\newcommand{\calA}{{\ensuremath{\cal A}}}
\newcommand{\path}{\ensuremath{\mbox{\it path}}}
\renewcommand{\path}{\ensuremath{\mbox{\it path}}}
\newcommand{\argmin}{\ensuremath{\mbox{argmin}}}
\newcommand{\argmax}{\ensuremath{\mbox{argmax}}}
\newcommand{\horSegment}{\rule[2pt]{30pt}{0.5pt}}

\newif\iffull
\fullfalse
\date{}
\title{EPG-rep\-re\-sen\-ta\-tions with small grid-size\thanks{Work done during
the 5th Workshop on Graphs and Geometry, Bellairs Research Institute.
The authors would like to thank the other participants, and especially
G\"{u}nter Rote, for helpful input.  Research of TB, VD and PM supported
by NSERC.  Research of MD supported by an NSERC Vanier scholarship.}
}
\author{Therese Biedl\inst{1}
\and 
Martin Derka\inst{1}
\and
Vida Dujmovic\inst{2}
\and
Pat Morin\inst{3}
}
\institute{Cheriton School of Computer
Science, Univ.~of Waterloo, Waterloo, Canada
\and
School of Computer Science and Electrical Engineering, Univ.~of Ottawa, Ottawa, Canada
\and
School of Computer Science, Carleton University, Ottawa, Canada}

\begin{document}
%\linenumbering
\maketitle
\begin{abstract}
In an EPG-representation of a graph $G$, each vertex is represented
by a path in the rectangular grid, and $(v,w)$ is an edge in $G$  if
and only if the paths representing $v$ an $w$ share a
grid-edge. Requiring paths representing edges to be x-monotone or, even stronger,
both x- and y-monotone gives rise to three natural variants of
EPG-rep\-re\-sen\-ta\-tions, one where edges have no monotonicity
requirements and two with the aforementioned monotonicity requirements.
The focus of this paper is understanding how small a grid can be achieved for such
EPG-rep\-re\-sen\-ta\-tions with respect to various graph parameters. 

We show that there are $m$-edge graphs that require a 
grid of area $\Omega(m)$ in any variant of
EPG-rep\-re\-sen\-ta\-tions. Similarly there are pathwidth-$k$ graphs that
require height $\Omega(k)$ and area $\Omega(kn)$  in any variant of
EPG-rep\-re\-sen\-ta\-tions. We prove a matching upper bound of $O(kn)$ area 
for all pathwidth-$k$ graphs in the strongest model, the one where edges are
required to be both x- and y-monotone. Thus in this strongest model, the result implies, for example, $O(n)$, $O(n
\log n)$ and $O(n^{3/2})$ area bounds for bounded pathwidth graphs, bounded
treewidth graphs and all classes of graphs that exclude a fixed minor,
respectively.
For the model with no restrictions on the monotonicity of the edges,
stronger results can be achieved for some graph classes, for example an $O(n)$ area bound for bounded treewidth graphs and $O(n \log^2 n)$
bound for graphs of bounded genus.
\end{abstract}

%%%%%%%%%%%%%%%%%%%%%%%%%%%%%%%%%%%%%%%%%%%%%%%%%%%%%%%%%%%%%%%%%%%%%%%%
\section{Introduction}\label{sec:intro}
The {\em $w\times h$-grid} (or {\em grid of width $w$ and height $h$})
consists of all {\em grid-points} $(i,j)$ that have integer coordinates $1\leq i\leq w$ and $1\leq j\leq h$, and all
{\em grid-edges} that connect grid-points of distance 1. An {\em
  EPG-rep\-re\-sen\-ta\-tion} of a graph $G$ consists of an
assignment of a {\em vertex-path}, $\path(v)$, to every vertex $v$ in
$G$ such that $\path(v)$ is a path in a grid, and $(v,w)$ is an edge
of $G$ if and only if $\path(v)$ and $\path(w)$ have a grid-edge in common.

Since their initial introduction by Go\-lum\-bic et
al.~\cite{GLS09-EPG}, a number of papers concerning EPG-rep\-re\-sen\-ta\-tions
of graphs have been published.  It is easy to see that {\em every} graph has an EPG-rep\-re\-sen\-ta\-tion \cite{GLS09-EPG}.
%Most of the subsequent papers either 
Later papers asked what graph classes can be represented if the number of bends in the
vertex-paths is restricted (see e.g.  \cite{AS09,BS09b,HKU14,FL16}) or
gave approximation algorithms for graphs with an EPG-rep\-re\-sen\-ta\-tion
with few bends (see e.g. \cite{EGM13,M17}).

The main objective of this paper is to find EPG-rep\-re\-sen\-ta\-tions such
that the size of the underlying grid is small (rather than the number of bends in vertex-paths).
%In addition, we would
%like, whenever possible, some ``nicer'' shapes for the vertex-paths.  Past
%papers usually restricted the number of bends.
%In contrast to this, we are interested
As done by Golumbic et al.~\cite{GLS09-EPG}, we wonder whether
additionally we can achieve monotonicity of vertex-paths.
We say that $\path(v)$ is {\em $x$-monotone} if any vertical line
that intersects $\path(v)$ intersects it in a single interval.
It is {\em $xy$-monotone} if it is $x$-monotone and additionally
any horizontal line that intersects $\path(v)$ intersects it in a single 
interval.  Finally, it is {\em $xy^+$-monotone} if it is monotonically
increasing, i.e., it is $xy$-monotone and the left endpoint is not
above the right endpoint.  
An {\em $x$-monotone EPG-rep\-re\-sen\-ta\-tion} is an EPG-rep\-re\-sen\-ta\-tion where every vertex-path is $x$-monotone, and similarly an {\em $xy^+$-monotone EPG-rep\-re\-sen\-ta\-tion} is an EPG-rep\-re\-sen\-ta\-tion where every vertex-path is $xy^+$-monotone.

%In this paper, we study what graph classes have EPG-rep\-re\-sen\-ta\-tions
%where the {\em size} of the grid (rather than the number of bends in
%vertex-paths) is small.  
It is easy to see that every $n$-vertex graph has an EPG-rep\-re\-sen\-ta\-tion
in an $O(n)\times O(n)$-grid, i.e., with quadratic area.  This is best
possible for some graphs. In Section~\ref{sec:lower}, we study lower bounds and show that there are $m$-edge graphs that require a 
grid of area $\Omega(m)$ in any EPG-rep\-re\-sen\-ta\-tion and that there are pathwidth-$k$ graphs that
require height $\Omega(k)$ and area $\Omega(kn)$  in any EPG-rep\-re\-sen\-ta\-tion. 

Biedl and Stern~ \cite{BS09b}  showed that pathwidth-$k$ graphs have an EPG-rep\-re\-sen\-ta\-tion
of height $k$ and width $O(n)$, thus area $O(kn)$.  In
Section~\ref{sec:pw}, we prove a strengthening of that result. In particular, we show that every
pathwidth-$k$ graph has an $xy^+$-monotone EPG-rep\-re\-sen\-ta\-tion of height $O(k)$
and width $O(n)$ thus
matching the lower bound in this strongest of the models.  This result implies, for example, $O(n)$, $O(n
\log n)$ and $O(n^{3/2})$ area bounds for  $xy^+$-monotone
EPG-rep\-re\-sen\-ta\-tions of bounded pathwidth graphs, bounded
treewidth graphs and all classes of graphs that exclude a minor,
respectively. In fact, the result implies that all hereditary graph
classes with $o(n)$-size balanced separators have $o(n^2)$ area $xy^+$-monotone EPG-rep\-re\-sen\-ta\-tions. 

If the monotonicity requirement is dropped, better area bounds
are possible for some graph classes. For example, in
Section~\ref{sec:epg}, we prove
that graphs of bounded treewidth have  $O(n)$ area EPG-rep\-re\-sen\-ta\-tions
and that graphs of bounded genus (thus planar graphs too) have $O(n
\log^2 n)$  EPG-rep\-re\-sen\-ta\-tions.

% TB: dropped this sentence for space reasons.
%We start with some graph theoretic notation in Section~\ref{sec:prelim} and
%introduce some useful tools in Section~\ref{sec:vpg}.

%%%%%%%%%%%%%%%%%%%%%%%%%%%%%%%%%%%%%%%%%%%%%%%%%%%%%%%%%%%%%%%%%%%%%%%%
\section{Preliminaries}\label{sec:prelim}

Throughout this paper, $G=(V,E)$ denotes a graph with $n$ vertices and
$m$ edges.  We refer, e.g., to \cite{Die12} for all standard notations for
graphs.  
The {\em pathwidth} $pw(G)$ of a graph $G$ is a well-known graph parameter.
Among the many equivalent definitions, we use here the 
following: $pw(G)$ is the smallest $k$ such that there exists
a super-graph $H$ of $G$ that is a $(k{+}1)$-colourable interval graph. 
Here, an {\em interval graph} is
a graph that has an {\em interval representation}, i.e., an assignment
of a (1-dimensional) interval to each vertex $v$ such that there exists an edge
if and only if the two intervals share a point.  
We may without loss of generality assume that the intervals begin
and end at distinct $x$-coordinates in $\{1,\dots,2n\}$, and will do
so for all interval representations used in this paper.
We use $I(v)=[\ell(v),r(v)]$ for the interval representing vertex $v$.
It is well-known that 
\iffull
interval graphs are {\em perfect} (see e.g. 
\cite{Gol04}), which in particular implies that 
\fi
an interval graph
is $k$-colourable if and only if its maximum clique size is $k$.

{\em Contracting} an edge $(v,w)$ of a graph $G$ means deleting 
both $v$ and
$w$, inserting a new vertex $x$, and making $x$ adjacent to all vertices
in $G-\{v,w\}$ that were adjacent to $v$ or $w$.    A graph $H$ is called
a {\em minor} of a graph $G$ if $H$ can be obtained from $G$ by deleting
some vertices and edges of $G$ and then contracting some edges of $G$.
It is known that $pw(H)\leq pw(G)$ for any minor $H$ of $G$.
 
\iffull
To give one example, for every interval graph we can easily find an
EPG-rep\-re\-sen\-ta\-tion (in a grid of height 1), simply by using $I(v)$
as the vertex-path for $v$.
\fi
% In this paper, we study EPG-rep\-re\-sen\-ta\-tions, but often obtain them
% by constructing a {\em VPG-representation} first. The latter consists
% of an assignment of vertex-paths in the grid to vertices
% such that $(v,w)$ is an edge
% if and only if $\path(v)$ and $\path(w)$ have a grid-point in common.

%%%%%%%%%%%%%%%%%%%%%%%%%%%%%%%%%%%%%%%%%%%%%%%%%%%%%%%%%%%%%%%%%%%%%%%%
\section{From proper VPG to EPG}\label{sec:vpg}

A  {\em VPG-representation} of a graph $G$ consists
of an assignment of vertex-paths in the grid to vertices of $G$ such
that $(v,w)$ is an edge of $G$ if and only if $\path(v)$ and $\path(w)$ have a grid-point in common.
Many previous EPG-rep\-re\-sen\-ta\-tion constructions (see e.g.~\cite{GLS09-EPG})
were obtained by starting with a VPG-representation and transforming
it into an EPG-rep\-re\-sen\-ta\-tion by adding a ``bump'' whenever two paths cross.
The lemmas below formalize this idea, and also study how this transformation 
affects the grid-size and whether monotonicity is preserved.  

We give the transformations only for {\em proper} VPG-representations, which
satisfy the following:  (a) Any grid-edge is used by at most one vertex-path.
%(2) there are no {\em knock-knees}, i.e., if two vertex-paths share a
%grid-vertex, then one vertex-path uses the two horizontal grid-edges and the
%other uses the two vertical grid-edges, 
(b) If a grid-point $p$ belongs to $\path(v)$ and $\path(w)$, then one
of the vertex-paths includes the rightward edge at $p$ and 
the other includes the upward edge at $p$.%
\footnote{The transformation could be done with a 
larger factor of increase if (b) is violated,
but restriction (a) is vital.
\iffull
$K_n$ has an EPG-rep\-re\-sen\-ta\-tion
of area 1, but its subgraph $K_{n/2,n/2}$ requires area $\Omega(n^2)$ as
argued in Corollary~\ref{cor:bipartite}.
\fi
} 
\iffull
Note that this condition
is satisfied if (a) holds and we ensure that there are no ``touchings'' 
of vertex-paths, i.e., any common grid-point implies a true intersection
where neither vertex-path ends.
\fi

\begin{lemma}
\label{lem:transform}
Let $G$ be a graph that has a proper VPG-representation $R_V$ 
in a $w\times h$-grid.
Then any subgraph $G'$ of $G$
has an EPG-rep\-re\-sen\-ta\-tion $R_E$ in a $2w\times 2h$-grid.
Furthermore, 
if $R_V$ is $x$-monotone then $R_E$ is $x$-monotone.
\end{lemma}
\begin{proof}
%Increase the resolution of the grid, by inserting a new row between any
%two consecutive rows, and two new columns between any two consecutive
%columns.  
Double the resolution of the grid by inserting a new grid-line after
each existing one.
For each edge $(v,w)$ of $G'$,
consider the two paths $\path(v)$ and $\path(w)$ that represent $v$
and $w$ in $R_V$.  Since $(v,w)$ was an edge of $G$, the vertex-paths 
share a grid-point $(i,j)$ in $R_V$, which corresponds to point %grid-point
$(2i,2j)$ in $R_E$.

Since $R_V$ is proper, we may assume  (after possible renaming)
that $\path(v)$ uses the rightward edge at $(2i,2j)$,  and
$\path(w)$ uses the upward edge at $(2i,2j)$.  
Re-route $\path(v)$ by adding a ``bump''
%$$(2i,2j)-(2i,2j{+1})-(2i{+}1,2j{+}1)-(2i{+}1,2j)$$
\begin{small}%
$$
\begin{array}{ccccccc}
    & (2i,2j{+}1) & \horSegment & (2i{+}1, 2j{+}1) \\
    & | &  		&  | \\
\horSegment & (2i,2j) & & (2i{+}1,2j) &  \horSegment \\
\end{array}
$$%
\end{small}%
in the first quadrant of $(2i,2j)$.  
See also Figure~\ref{fig:transform1}(a) and (b).
Note that $\path(w)$ is unchanged in the vicinity of
$(2i,2j)$, and the bump added to $\path(v)$ is $x$-monotone.
So if $R_V$ is $x$-monotone then so are the resulting vertex-paths.

Since $R_V$ is proper, no
other vertex-path used $(i,j)$ in $R_V$, and therefore no
other vertex-path in $R_E$ can use any grid-edge of this bump.
Therefore no new adjacencies are created, and $R_E$ is indeed
an EPG-rep\-re\-sen\-ta\-tion of $G'$.\qed
\end{proof}

\begin{figure}[ht]
\hspace*{\fill}
\begin{subfigure}[b]{0.22\linewidth}
	\includegraphics[height=16mm,page=1]{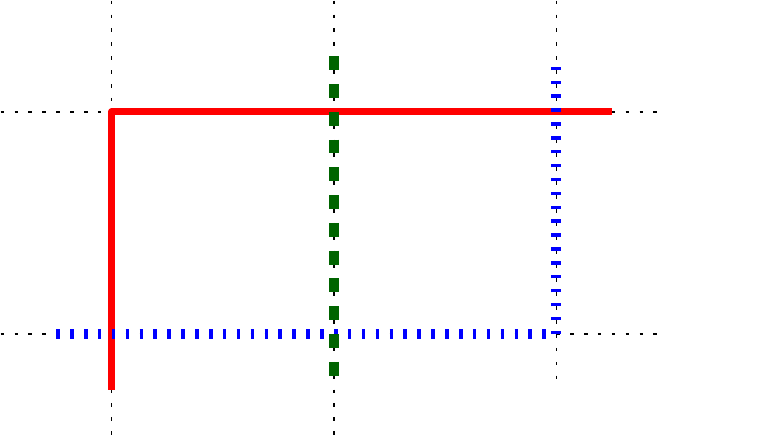}
	\caption{VPG.}
\end{subfigure}
\hspace*{\fill}
\begin{subfigure}[b]{0.22\linewidth}
	\includegraphics[height=16mm,page=2]{transform.pdf}
	\caption{EPG.}
\end{subfigure}
\hspace*{\fill}
%\newline
%\hspace*{\fill}
\begin{subfigure}[b]{0.22\linewidth}
	\includegraphics[height=16mm,page=3]{transform.pdf}
	\caption{Skewed VPG.}
\end{subfigure}
\hspace*{\fill}
\begin{subfigure}[b]{0.24\linewidth}
	\includegraphics[height=16mm,page=4]{transform.pdf}
	\caption{$xy^+$-mon.~EPG.}
\end{subfigure}
\hspace*{\fill}
\caption{Transforming a proper VPG-representation. 
We only show the transformation for the edge from blue (dotted) to
green (dashed) vertex.
}
\label{fig:transform1}
\end{figure}

We now give a second construction, which is similar in spirit, but
re-routes differently in order to preserve $xy^+$-monotonicity.

\begin{lemma}
\label{lem:transformxyplus}
\label{lem:transformXYmonotone}
Let $G$ be a graph that has a proper VPG-representation $R_V$ 
in a $w\times h$-grid with $xy^+$-monotone vertex-paths.
Then any subgraph $G'$ of $G$
has an $xy^+$-monotone EPG-rep\-re\-sen\-ta\-tion $R_E$ in a $(2w+h)\times 2h$-grid
\end{lemma}
\begin{proof}
We do two transformations; the first results in a proper VPG-representation
$R_V'$ that has some special properties such that it can then be transformed
into an EPG-rep\-re\-sen\-ta\-tion.

The first transformation is essentially a skew.  Map each grid-point $(i,j)$
of $R_V$ into the {\em corresponding point} $(2i+j,2j)$.  Any horizontal
grid-edge used by a vertex-path is mapped to the corresponding 
horizontal grid-edge, i.e., we map a horizontal grid-edge $(i,j)-(i{+}1,j)$
of $R_V$  into the length-2 horizontal segment $(2i+j,2j)-(2(i+1)+j,2j)$
that connects the corresponding points. 
Every vertical grid-edge 
%used by a vertex-path is mapped to a zig-zag
%connecting the corresponding grid-points.  Thus, we map a vertical
%grid-edge 
$(i,j)-(i,j{+}1)$ is mapped into the zig-zag path
%$$ (2i{+}j,2j)-(2i{+}j,2j{+}1)-(2i{+}j{+}1),2j{+}1)-(2i{+}(j{+}1),2(j{+}1)).$$
\begin{small}
$$
\begin{array}{ccccccc}
    &			&	& (2i{+}(j{+}1), 2(j{+}1) \\
&		&	& | \\
    & (2i{+}j,2j{+}1) & \horSegment & (2i{+}j{+}1, 2j{+}1) \\
    & | &  		\\
    & (2i{+}j,2j) & \\
\end{array}
$$%
\end{small}%
that connects the corresponding points.    See also 
Figure~\ref{fig:transform1}(a) and (c).
It is easy to verify that this is again a proper VPG-representation of
exactly the same graph, and vertex-paths are again $xy^+$-monotone.

Now view $R_V'$ as an EPG-rep\-re\-sen\-ta\-tion.  Since $R_V'$ is proper, currently
no edge is represented.  We now modify $R_V'$ such that intersections are
created if and only if an edge exists.  Consider some
edge $(v,w)$ of $G'$.  Since it is an edge of $G$, there must exist a point
$(i,j)$ in $R_V$ where $\path(v)$ and $\path(w)$ meet.  
Since $R_V$ is proper, we may assume  (after possible renaming)
that $\path(v)$ uses the rightward edge at $(i,j)$ while
$\path(w)$ uses the upward edge at $(i,j)$.  Consider the corresponding
point $(2i+j,2j)$ in $R_V'$, and observe that $\path'(v)$ and $\path'(w)$
(the vertex-paths in $R_V'$) use its incident rightward and upward edges, respectively.
Moreover, $\path'(w)$ uses the ``zig-zag'' $(2i+j,2j)-(2i+j,2j+1)-(2i+j+1,2j+1)$.
We can now re-route the vertex-path of $w$ to use instead
$(2i+j,2j)-(2i+j+1,2j)-(2i+j+1,2j+1)$, i.e., to share the horizontal edge
with $\path'(w)$ and then go vertically.  
See Figure~\ref{fig:transform1}(d).
Thus the two paths now share a
grid-edge.  Since no other vertex-paths used $(i,j)$ in $R_V$,  this
re-routing does not affect any other intersections and overlaps.    So
we obtain an EPG-rep\-re\-sen\-ta\-tion of $G$, and one easily verifies that it
is $xy^+$-monotone. \qed
\end{proof}

\begin{theorem}
Every graph $G$ 
with $n$ vertices has an $xy^+$-monotone EPG-re\-pre\-sen\-ta\-tion 
in a $3n\times 2n$-grid.
\end{theorem}
\begin{proof}
It is very easy to create a proper VPG-representation of the complete
graph $K_n$ in an $n\times n$-grid, using a $\Gamma$-shape (hence an
$xy^+$-monotone vertex-path). Namely, place the corner of the $\Gamma$ of
vertex $i$ at $(i{-}1,i)$ and extending the two arms to $y=1$ and $x=n$.
For vertex $1$, the grid-edge $(0,1)-(1,1)$ is not needed and
can be omitted to save a column.
See Figure~\ref{fig:completeGraph}.
Since $G$ is a subgraph of $K_n$, the result then follows by
Lemma~\ref{lem:transformxyplus}.\qed
\begin{figure}[ht]
\hspace*{\fill}
\includegraphics[width=0.35\linewidth,page=3,trim=0 0 200 50,clip]{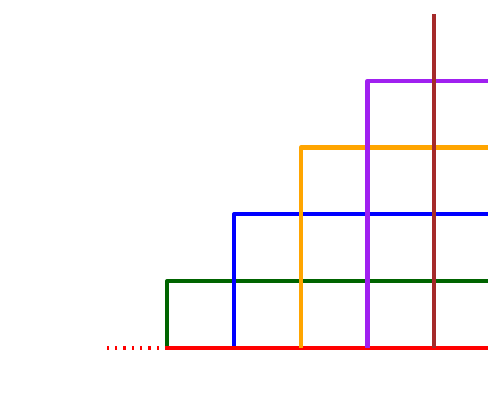}
\hspace*{\fill}
\includegraphics[width=0.55\linewidth,page=4,trim=0 0 120 50,clip]{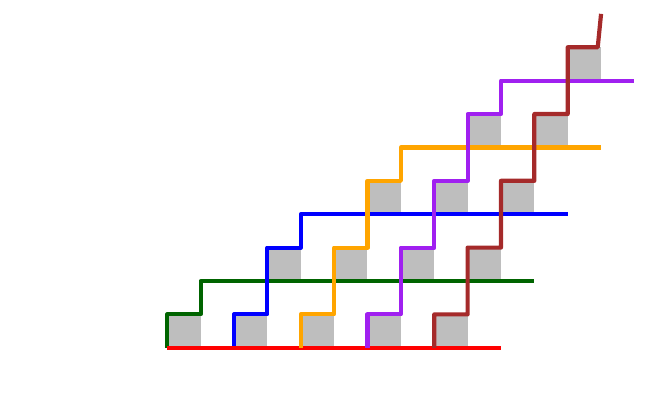}
\hspace*{\fill}
\caption{A VPG-representation of $K_n$, and an EPG-rep\-re\-sen\-ta\-tion for
any graph.  In gray areas vertex-paths may get re-routed
to create shared grid-edges.}
\label{fig:completeGraph}
\end{figure}
\end{proof}

Contrasting this with existing results, 
%This transformation gives an 
%$xy^+$-monotone EPG-rep\-re\-sen\-ta\-tion for any graph.
it was already known that any graph has an EPG-rep\-re\-sen\-ta\-tion
\cite{GLS09-EPG}, but our construction additionally imposes
$xy^+$-monotonicity, and our grid-size is $O(n^2)$, rather
than $O(nm)$.

%%%%%%%%%%%%%%%%%%%%%%%%%%%%%%%%%%%%%%%%%%%%%%%%%%%%%%%%%%%%%%%%%%%%%%%%
\section{Lower bounds}\label{sec:lower}

We now turn to lower bounds.  These
hold for arbitrary EPG-rep\-re\-sen\-ta\-tions; we make no use
of monotonicity.

\begin{theorem}\label{thm:numedg}
Let $G$ be a triangle-free graph with $m$ edges.  Then any EPG-rep\-re\-sen\-ta\-tion
of $G$ uses at least $m$ grid-edges (hence a grid of area $\Omega(m)$).
\end{theorem}
\begin{proof}
If $G$ has no triangle, then the maximal clique-size is 2.  Hence no 
grid-edge can belong to three or more vertex-paths.  Consequently,
for every edge $(v,w)$ we must have at least one grid-edge (the one that
is common to $\path(v)$ and $\path(w)$).   No grid-edge belongs to three
vertex-paths, and so there must be at least $m$ grid-edges.\qed
\end{proof}
A consequence of Theorem~\ref{thm:numedg} is that $K_{n,n}$ requires
$\Omega(n^2)$ area in any EPG-rep\-re\-sen\-ta\-tion. 
Later, we relate pathwidth to EPG-rep\-re\-sen\-ta\-tions. For now, we note that $K_{n-k,k}$ is an $n$-vertex triangle-free graph with pathwidth $k$ and $\Theta(kn)$ edges.  Together with Theorem~\ref{thm:numedg}, this implies:

\begin{corollary}
\label{cor:pathwidth-area}
For every $k\ge 1$ and every $n\ge 2k$, there exists an $n$-vertex
pathwidth-$k$ graph $G$ for which any EPG-rep\-re\-sen\-ta\-tion of $G$ uses
$\Omega(kn)$ grid-edges (hence a grid of area $\Omega(kn)$).
\end{corollary}

One wonders whether there are graphs that have only a linear number of
edges and still require a big, even quadratic, area.  The following
lower bound, also based on pathwidth, allows us to answer this question
in the affirmative.

\begin{theorem}
\label{thm:pathwidth}
Let $G$ be a 
graph that has an EPG-rep\-re\-sen\-ta\-tion in a grid with $h$ rows
and for which any grid-edge is used by at most $c$ vertex-paths.
Then $pw(G)\leq c(3h-1)-1$.
\end{theorem}
\begin{proof}
For every vertex $v$, define $I(v)$ to be the $x$-projection
of $\path(v)$.  This is
an interval since $\path(v)$ is connected.  Define $H$ to be the interval
graph of these intervals.
If $(v,w)$ is an edge,
then $\path(v)$ and $\path(w)$ share a grid-edge, and hence the 
intervals $I(v)$ and $I(w)$
share at least one point.  So $G$ is a subgraph of $H$.
We claim that $H$ has clique-size $\omega(H)\leq 6h-2$; this implies
the result.
% since $H$ then is $(4h{-}2)$-colourable.  

Fix an arbitrary
maximal clique $D$ in $H$.  It is well-known
(see e.g.~\cite{Gol04}) that, in the projected interval-representation,
there exists a vertex $v$ such that $D$ corresponds to those vertices whose
intervals intersect the left endpoint $\ell(v)$.  
Hence for any vertex $w$ in $D$,
at least one grid-edge of $\path(w)$ is incident to a grid-point with 
$x$-coordinate $\ell(v)$.
There are only $3h-1$ such grid-edges ($2h$ horizontal ones and $h-1$
vertical ones), and each of them can belong to at most $c$ vertex-paths.
Hence $|D|\leq c(3h-1)$, which proves the claim.\qed
\end{proof}

In particular, if $G$ is triangle-free then no three vertex-paths can
share a grid-edge.  Applying the theorem with $c=2$ for such graphs we get:
\begin{corollary}
\label{cor:pathwidth}
Any triangle-free graph with pathwidth $k$ 
requires an $\Omega(k)\times \Omega(k)$-grid and thus $\Omega(k^2)$
area in any
EPG-rep\-re\-sen\-ta\-tion.
\end{corollary}

So all that remains to do for a better lower bound is to find a graph that
has few edges yet high pathwidth.  For this, we use {\em expander-graphs},
which are graphs such that for any vertex-set $S$ the ratio between the
boundary of $S$ (the number of vertices in $S$ with neighbors in $V-S$)
and $|S|$ is bounded from below.
%We will not give their complicated definition (see for example \cite{GM09});
%roughly speaking these are graphs that cannot be separated in a balanced
%way using $o(n)$ vertices.  

\begin{theorem}
There are $n$-vertex graphs with $O(n)$ edges for which any EPG-rep\-re\-sen\-ta\-tion
requires $\Omega(n^2)$ area.
\end{theorem}
\begin{proof}
It is known that expander-graphs of maximum degree 3 exist (see e.g.~\cite{MSS13})
Let $G$ be one such graph. It hence has $O(n)$ edges. Since $G$ is an
expander, it has pathwidth $\Omega(n)$ (see e.g. \cite{GM09}).
Subdivide all edges of $G$ to obtain a bipartite graph $G'$ that has
$O(n)$ vertices and edges. This operation
cannot decrease the pathwidth since $G$ is a minor of $G'$.
So $pw(G')\in \Omega(n)$.  Since $G'$ is triangle-free, any EPG-rep\-re\-sen\-ta\-tion
of $G'$ must have height $\Omega(n)$, and a symmetric argument shows that
it must have width $\Omega(n)$.\qed
\end{proof}

%%%%%%%%%%%%%%%%%%%%%%%%%%%%%%%%%%%%%%%%%%%%%%%%%%%%%%%%%%%%%%%%%%%%%%%%
%\section{Constructions for small pathwidth}\label{sec:pw}
\section{Upper bounds on $xy^+$-monotone EPG representations}\label{sec:pw}

Corollaries~\ref{cor:pathwidth-area} and \ref{cor:pathwidth}
imply that the best upper-bounds for EPG-rep\-re\-sen\-ta\-tions
in terms of pathwidth have height $\Omega(k)$ and area $\Omega(kn)$.
Naturally, one wonders whether this bound can be matched.
As noted in the introduction, Biedl
and Stern showed that any graph with pathwidth $k$ has an EPG-rep\-re\-sen\-ta\-tion
of height $k$ and area $O(kn)$ \cite{BS09b}. 
We now use a completely different approach to strengthen their result
and obtain $xy^+$-monotone EPG-rep\-re\-sen\-ta\-tions of pathwidth $k$ graphs
with optimal height $O(k)$ and optimal area $O(kn)$.

\begin{theorem}
\label{thm:pwXYmonotone}
Every graph $G$ of pathwidth $k$
has an $xy^+$-monotone EPG-rep\-re\-sen\-ta\-tion 
of height $8k+O(1)$ and width $O(n)$, thus with $O(kn)$ area.
\iffalse
\footnote{We did not minimize the `$+O(1)$' to keep the
presentation simple.  It is possible to obtain a bound of 
%$2k$ for Theorem~\ref{thm:pwXmonotone} and $8k$ for Theorem~\ref{thm:pwXYmonotone}
$8k$
but this involves using an equivalent pathwidth definition based on
an elimination order, and doing a special handling of the edge between
the vertex that is newly added and the vertex that ceases being active.
(\cite{BS09b} used this approach.)  We leave details to the reader.
}
\fi
\end{theorem}
\begin{proof}
Recall that $G$ is a subgraph of a
$(k{+}1)$-colourable interval graph $H$.
By Lemma~\ref{lem:transformxyplus}, it suffices to 
show the following:

\begin{lemma}
Let $H$ be a $(k{+}1)$-colourable interval graph
with interval representation $\{I(v)=[\ell(v),r(v)] : v\in V\}$.
There exists a proper VPG-representation 
with $xy^+$-monotone vertex-paths 
of a supergraph of $H$ such that
\begin{enumerate}
\item all vertex-paths are contained within the $[2,2n+1]\times [-2k-2,2k+1]$-grid
	(more precisely, the $x$-range is $[2\min_{v\in V} \ell(v), 1+2 \max_{v\in V} r(v)]$); 
%the grid $[2\min_{v\in V} \ell(v), 1+2 \max_{v\in V} r(v)] \times [-2k-2,2k+1]$
\item $\path(v)$ contains a horizontal segment whose $x$-range is 
	$[2\ell(v),2r(v)]$ and whose $y$-coordinate is negative; and
\item some vertical segment $\path(v)$ includes the segment 
	$\{2r(v)\} \times [-1,1]$.
%\item rows $0$ and $2k-1$ contain no horizontal segments of vertex-paths.
\end{enumerate}
\end{lemma}

We prove the lemma by induction on $k$.
We may 
assume that $H$ is connected, for if it is not, then obtain representations
of each connected component separately and combine them.  The $x$-ranges
of intervals of each component are disjoint (else there would be an edge), and
so the representations of the components do not overlap by (1).

The claim is straightforward for $k=0$:  Since $H$ is connected and 
1-colourable, it has only one vertex $v$. Set $\path(v)$ to use the two
segments $[2\ell(v),2r(v)]\times \{-1\}$ 
and $\{2r(v)\}\times [-1,1]$.  All claims hold.

Now assume that $k\geq 1$.
We find a path $P$ of ``farthest-reaching'' intervals as follows.
Set $a_1:=\argmin_{v\in V} \ell(v)$, i.e., $a_1$ is the interval that
starts leftmost.  Assume $a_i$ has been defined for some $i\geq 1$.  Let
$\calA_i$ be the set of all vertices $v$ with $\ell(a_i)<\ell(v)<r(a_i)<r(v)$.
If $\calA_i$ is empty then stop the process; we have reached the last vertex
of $P$.  Else, set $a_{i+1}:=\argmax_{v\in \calA_i} r(v)$ to be the vertex
in $\calA_i$ whose interval goes farthest to the right, and
repeat.    See also Figure~\ref{fig:choosePath}.
Let $P=a_1,a_2,\dots,a_p$ be the path that we obtained (this is indeed
a path since $I(a_i)$ intersects $I(a_{i+1})$ by definition).

\begin{figure}[ht]
\hspace*{\fill}
\includegraphics[width=0.8\linewidth,page=3,trim=20 40 20 90,clip]{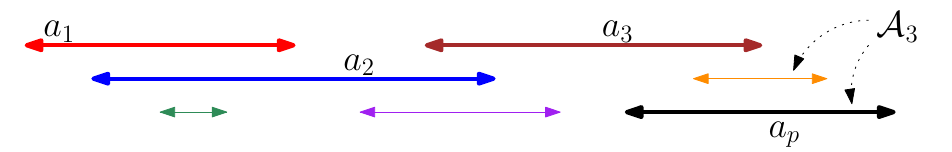}
\hspace*{\fill}
\newline
\newline
\hspace*{\fill}
\includegraphics[width=0.8\linewidth,page=4,trim=20 40 20 90,clip]{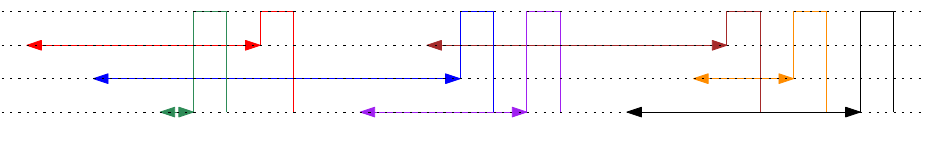}
\hspace*{\fill}
\caption{An interval graph (bold intervals denote the path $P$ chosen in Theorem~\ref{thm:pwXYmonotone}), and its proper VPG-representation with $x$-monotone vertex-paths.}
\label{fig:choosePath}
\end{figure}

\begin{claim} $P$ is an induced path. \end{claim}
\begin{proof}
It suffices to show that $r(a_i)<\ell(a_{i+2})$
for all $1\leq i\leq p-2$.  Assume for contradiction that $\ell(a_{i+2})
< r(a_i)$ for some $1\leq i\leq p-2$.  We show this contradicts the choice
of $P$ as the vertices that go farthest right.  Namely,
let $j\leq i$ be the smallest index such that $\ell(a_{i+2})
< r(a_j)$.  
If $j>1$ then $\ell(a_{i+2})>r(a_{j-1})>\ell(a_j)$ by
definition of $j$ and ${\cal A}_{j-1}$.  If
$j=1$ then $\ell(a_{i+2}) \geq \min_{v\in V} \ell(v) = \ell(a_1)=\ell(a_j)$, and
the inequality is strict since $i+2\neq 1$.    
Thus in both cases $\ell(a_{i+2})>\ell(a_j)$.  Therefore 
$\ell(a_j)<\ell(a_{i+2})<r(a_j)\leq r(a_{i+1})<r(a_{i+2})$, which implies
$a_{i+2}\in \calA_j$.  By $r(a_{j+1})\leq r(a_{i+1})<r(a_{i+2})$ this 
contradicts the
choice of $a_{j+1}$ as $\argmax_{v\in \calA_j} r(v)$.
\qed
\end{proof}

By definition $a_1$ is the leftmost
interval, i.e., $\ell(a_1)=\min_{v\in V} \ell(v)$.  We claim that $a_p$ is
the rightmost interval, i.e., $r(a_p)=\max_{v\in V} r(v)$. 
Assume for contradiction
that some vertex $v$ has an interval that ends farther right.  By connectivity
we can choose $v$ so that it intersects $I(a_p)$, thus $\ell(v)<r(a_p)<r(v)$.
Let $j\leq p$ be maximal such that $\ell(v)<r(a_j)$.  Similarly, as in the claim, one
argues that $v\in \calA_j$, and therefore $v$, rather than $a_{j+1}$, should
have been added to path $P$.

We are now ready for the construction.  Define $H':=H-P$.  Since the intervals
of $P$ cover the entire range $[\min_{v\in V} \ell(v),\max_{v\in V} r(v)]$,
any maximal clique of $H$ contains a vertex of $P$.  Therefore the maximum
clique-size of $H'$ satisfies $\omega(H')\leq \omega(H)-1$, which implies 
(for an interval-graph) that $\chi(H')\leq \chi(H)-1$, hence $H'$ is $k$-colourable.  Apply induction to
$H'$ (with the induced interval representation)
and let $\Gamma'$ be the resulting VPG-representation.    

Since $\Gamma'$ uses only orthogonal vertex-paths, we can 
insert two rows each above and below the $x$-axis by moving all other bends
up/down appropriately.  Now set $\path(a_i)$ to be
\begin{small}
$$
\begin{array}{ccccccc}
& & (2r(a_i),-Y) & \horSegment & (2r(a_{i+1})+1,Y) \\
& & | \\
 (2\ell(a_i),-Y) & \horSegment & (2r(a_i),-Y) \\
| \\
(2\ell(a_i),-2k-2) & & \\
\end{array}
$$
\end{small}%
where $Y=1$ if $i$ is odd and $Y=2$ if $i$ is even.    We omit the 
rightmost horizontal segment for $i=p$  (because $a_{p+1}$ is undefined).
See also Figure~\ref{fig:createCurve}.

\begin{figure}[ht]
\hspace*{\fill}
\includegraphics[width=\linewidth,page=5]{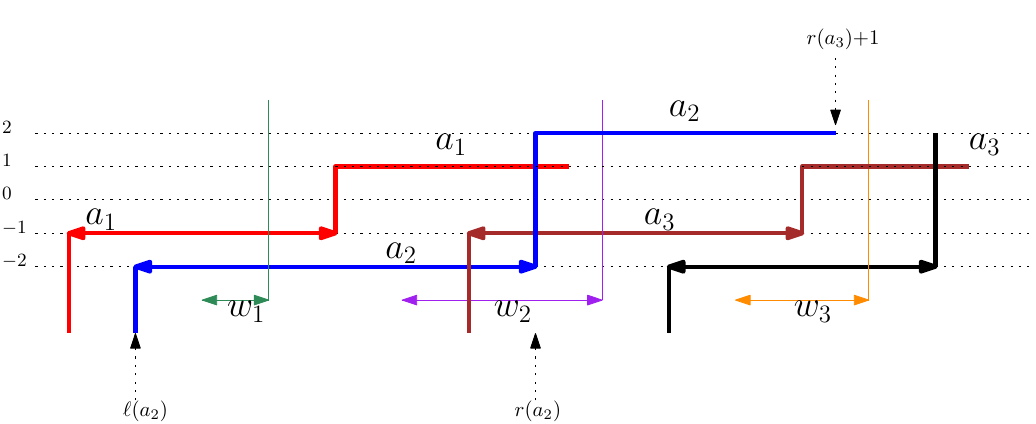}
\hspace*{\fill}
\caption{Representation with $xy$-monotone vertex-paths.}
\label{fig:createCurve}
\end{figure}

Note that these vertex-paths satisfy conditions (2) and (3).  Also note that for
any $1\leq i<p$, the vertex-paths of $a_i$ and $a_{i+1}$ intersect, namely at
$(2r(a_{i+1}),1)$ if $i$ is odd and at $(2\ell(a_{i+1}),-2)$ and $(2r(a_i),-1)$
if $i$ is even.  It remains to show that for any edge $(w,a_i)$ (for some
$1 \leq i \leq p$ and $w\not\in P$) the vertex-paths intersect.  Here we have
three cases (Figure~\ref{fig:createCurve} illustrates the $\ell$th case for edge 
$(w_\ell,a_{\ell+1})$):
\begin{enumerate}
\item If $r(w)<r(a_i)$, then $\ell(a_i)<r(w)$, else the intervals
	would not intersect.    By (3), and since we inserted new rows
	around the $x$-axis, we know that $\path(w)$ contains
	the vertical segment $2r(w) \times [-3,3]$.
	Therefore $\path(a_i)$ intersects this segment at $(2r(w),-Y)$
	where $Y\in \{1,2\}$. 
\item If $\ell(w)<\ell(a_i)$, then $\ell(a_i)<r(w)$, else the intervals
	would not intersect.    By (2), and since we inserted new rows
	around the $x$-axis, we know that $\path(w)$ has
	a horizontal segment $[2\ell(w),2r(w)]\times {-}Y$ for some $Y\geq 3$.  
	Therefore $\path(a_i)$ intersects this segment at $(2\ell(a_i),-Y)$.
\item Finally assume that $\ell(a_i)<\ell(w)$ and $r(a_i)<r(w)$.
	We must have $\ell(w)<r(a_i)$, else the intervals would not
	intersect.  Therefore $w\in \calA_i$. By choice of $a_{i+1}$ 
	we have
	$r(w)\leq \max_{v\in \calA_i} r(v) = r(a_{i+1})$.  
	By (3), and since we inserted new rows
	around the $x$-axis, we know that $\path(w)$ contains
	the vertical segment $2r(w) \times [-3,3]$.  By $r(w)\leq r(a_{i+1})$,
	therefore $\path(a_i)$ intersects this segment at $(2r(w),Y)$
	where $Y\in \{1,2\}$. 
\end{enumerate}
Hence all edges of $H$ are represented by intersection of vertex-paths
and as one easily verifies, these are proper intersections.
This finishes the induction and proves the theorem.
%\smallskip
%
%The construction as described has at most $4k+3$ rows.
%We used row $0$ only to have more symmetric $y$-coordinates, so it can
%be deleted to give $4k+2$ rows. 
%Apply the transformation of Lemma~\ref{lem:transformXYmonotone}; the
%result uses $8k+4$ rows.
%Recall that row $2k{+}1$ of the VPG-representation
%had no horizontal grid-edges of vertex-paths in it.
%A closer inspection of the transformation shows the two rows that correspond
%to it will only be used for vertex-paths that end here, and no grid-edge 
%incident to them is used by two vertex-paths.  They therefore also can be
%deleted, and we end with an EPG-rep\-re\-sen\-ta\-tion on $8k+2$ rows.
\qed
\end{proof}

We can use Theorem~\ref{thm:pwXYmonotone}
to obtain small EPG-representations for other
%These facts together with Theorem \ref{thm:pwXYmonotone} imply the following results:
graph classes.
% TB: resorted to match order of claims below.
Graphs of bounded treewidth have pathwidth at most $O(\log n)$ \cite{Bodlaender:1998:PKA:292475.292476}. 
%Alon, Seymour, and Thomas \cite{Alon:1990:STG:100216.100254} proved that
Graphs excluding a fixed minor have treewidth
$O(\sqrt{n})$ \cite{Alon:1990:STG:100216.100254}. 
A graph class has
treewidth $O(n^\epsilon)$
if it is {\em hereditary} (subgraphs also
belong to the class) and has {\em balanced separators} (for any
weight-function on the vertices there exists a small set $S$ such that
removing $S$ leaves only components with at most half the weight)
for which the {\em size} (the cardinality of $S$) is at most
$O(n^{\epsilon})$, for some fixed $\epsilon\in(0,1)$
\cite{DBLP:journals/corr/DvorakN14}.
It is well known that hereditary graph classes with
treewidth $O(n^{\epsilon})$, for some fixed $\epsilon\in(0,1)$, have pathwidth $O(n^\epsilon)$ (see
\cite{Bodlaender:1998:PKA:292475.292476} for example),
so graphs excluding a fixed minor have pathwidth $O(\sqrt{n})$ and
graphs with $O(n^\epsilon)$-size balanced separators have pathwidth
$O(n^\epsilon)$.  This implies:

%\begin{corollary}
\begin{itemize}
%\item Graphs of bounded pathwidth have $xy^+$-monotone
% EPG-rep\-re\-sen\-ta\-tions with $O(n)$ area. 
\item Graphs of bounded treewidth have $xy^+$-monotone
 EPG-rep\-re\-sen\-ta\-tions %with $O(n \log n)$ area.
in an $O(\log n)\times O(n)$-grid.
\item Graphs excluding a fixed minor have  $xy^+$-monotone
 EPG-rep\-re\-sen\-ta\-tions in an $O(\sqrt{n})\times O(n)$-grid.
%, and thus $O(n^{3/2}$) area. 
%This includes planar graphs and more generally all bounded genus graphs.
\item Hereditary classes of graphs with
%that have a balanced separator of size
$O(n^\epsilon)$-sized balanced separators for some $\epsilon\in(0,1]$ have  $xy^+$-monotone EPG-rep\-re\-sen\-ta\-tions
%in $O(n^{1+\epsilon})$ area.
in an $O(n^\epsilon)\times O(n)$-grid.
\end{itemize}
%\end{corollary}

The $O(n)$ area result for bounded pathwidth graphs is tight by Theorem
\ref{thm:numedg}. Naturally, one wonders if the other three results in
above are tight. The $O(\sqrt{n})\times O(n)$-grid
%($O(n^{3/2}$) area) 
bound applies, for example, to all planar graphs
and more generally all bounded genus graphs.  Some planar graphs with $n$ vertices have pathwidth $\Omega(\sqrt{n})$
(the $\sqrt{n}\times \sqrt{n}$-grid is one example), so the height cannot
be improved. 
But can the width or the area be improved? 
This turns out to be true for some graph classes
if the monotonicity condition is dropped.  In the next section,
we show these improved bounds
via a detour into orthogonal drawings.
Some of these results are tight.

%. In the next section, we
%derive improved bounds for some graph classes, some of which are tight
%in fact.
%For such non-monotone
%EPG-rep\-re\-sen\-ta\-tions we also derive tight area bound for bounded treewidth graphs.

%%%%%%%%%%%%%%%%%%%%%%%%%%%%%%%%%%%%%%%%%%%%%%%%%%%%%%%%%%%%%%%%%%%%%%%%
\section{EPG-rep\-re\-sen\-ta\-tions via orthogonal drawings}\label{sec:epg}

In this section, we study another method of obtaining EPG-rep\-re\-sen\-ta\-tions,
which gives (for some graph classes) even smaller EPG-rep\-re\-sen\-ta\-tions.
Define a {\em 4-graph} to be a graph where all vertices have degree at most
4.  An {\em orthogonal drawing} of a 4-graph is an assignment of grid-points
to vertices and grid-paths to edges such that the path of each edge connects
the grid points of its end-vertices.  Edges are allowed to intersect, but any such
intersection point must be a true intersection, i.e., one edge uses only
horizontal grid-edges while the other uses only vertical grid-edges at
the intersection point.

\begin{lemma}
\label{lem:orthTransform}
Let $G$ be a 4-graph that has an orthogonal drawing in a $w\times h$-grid.
Then any minor of $G$ has an EPG-rep\-re\-sen\-ta\-tion in a $2w\times 2h$-grid.
\end{lemma}
\begin{proof}
First delete from the orthogonal drawing all edges of $G$ that are not needed 
for the minor $H$; this cannot increase the grid-size.  So we may assume
that $H$ is obtained from $G$ via edge contractions only.

We first explain how to obtain an EPG-rep\-re\-sen\-ta\-tion of $G$.
Double the grid by inserting a new row/column after each existing one.
Every grid-point that belonged to a vertex $v$ hence now corresponds to 4
grid-points that form a unit square; denote this by $\square_v$.
Duplicate all segments of grid-paths for edges in the adjacent new grid-line, 
and extend/shorten suitably so that the copies again form grid-paths, 
connecting the squares of their end.  Thus for each edge $(v,w)$ we now
have two grid-paths $P_{v,w}^1$ and $P_{v,w}^2$ from $\square_v$ to 
$\square_w$.

We now define $\path(v)$ (which will be a closed path) by tracing
the edges of the orthogonal drawing suitably.  To describe this in more
detail, first arbitrarily direct the edges of $G$.    
Initially, $\path(v)$ is simply the boundary of $\square_v$.  Now
consider each edge $(v,w)$ incident to $v$.  If it is directed $v\rightarrow w$,
then remove from $\path(v)$ the grid-edge 
along $\square_v$ that connects the two ends of
$P_{v,w}^1$ and $P_{v,w}^2$, add these two grid-paths, and add the
grid-edge $e'$ along $\square_w$ that connects these two paths. 
Note that $e'$ also belongs to $\path(w)$, so with this $\path(v)$ and 
$\path(w)$ share a grid-edge and we obtain the desired EPG-rep\-re\-sen\-ta\-tion
of $G$.

\begin{figure}[ht]
\hspace*{\fill}
\begin{subfigure}[b]{0.31\linewidth}
	\includegraphics[width=\linewidth,page=1]{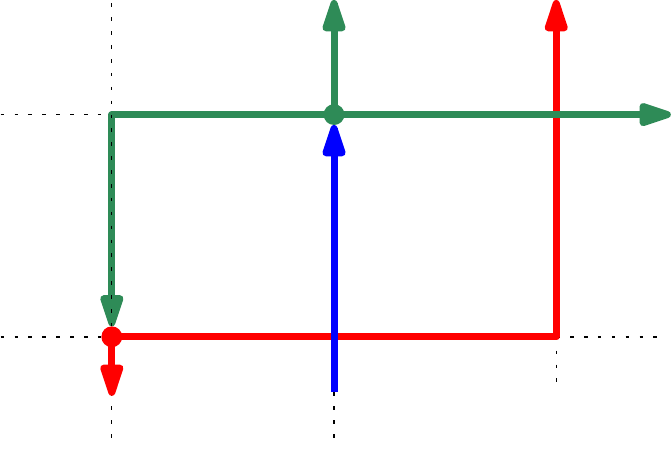}
	\caption{}
\end{subfigure}
\hspace*{\fill}
\begin{subfigure}[b]{0.31\linewidth}
	\includegraphics[width=\linewidth,page=2]{orthTransform.pdf}
	\caption{}
\end{subfigure}
\hspace*{\fill}
%\newline
%\hspace*{\fill}
\begin{subfigure}[b]{0.31\linewidth}
	\includegraphics[width=\linewidth,page=3]{orthTransform.pdf}
	\caption{}
\end{subfigure}
\hspace*{\fill}
\caption{Transforming an orthogonal drawing into an EPG-rep\-re\-sen\-ta\-tion.
For ease of reading we show the duplicated grid-line close to the original one.}
\label{fig:transform}
\end{figure}

It remains to argue that this can be turned into an EPG-rep\-re\-sen\-ta\-tion
of a graph $H$ obtained from $G$ via edge contractions.  Suppose we want to
contract edge $(v,w)$.  The two grid-paths $\path(v)$ and $\path(w)$ share a grid-edge $e$
that belongs to no other vertex-path.
Delete $e$ from both paths, and let the path of the contraction-vertex
be the union of the two 
resulting open paths, which is again a closed path.  Thus we obtain
an EPG-rep\-re\-sen\-ta\-tion of $H$ where all vertex-paths are closed paths.

If desired, we can turn this into an EPG-rep\-re\-sen\-ta\-tion with open paths
by deleting for every $v\in V$ one grid-edge from $\path(v)$
that is not shared with any other vertex-path.  If $\deg(v)\leq 3$, then
a suitable edge is the grid-edge of $\square_v$ on the side where no edge 
attaches.  If $v$ has an outgoing edge $v\rightarrow w$, then a suitable
edge is any grid-edge of $P_{v,w}^1$.    We can achieve that one of these
always holds as follows:
%by directing suitably.  Namely,
If all vertex degrees of $G$ are 4, then direct $G$ by walking along
an Eulerian cycle; then all vertices have outgoing edges.  
If some vertex $v$ has degree 3 or less, then find
a spanning tree $T$ of $G$, root it at $v$, direct all tree-edges towards
the root and all other arbitrarily.  Either way, this direction satisfies
that any vertex of degree 4 has at least one outgoing edge and we can
delete an edge of each $\path(v)$ such that all vertex-paths are open paths. \qed
\end{proof}

We note that a somewhat similar transformation from orthogonal drawings
was used recently to create pixel-representations \cite{ABR+15}, but in
contrast to their result we do not need the orthogonal drawings to be
planar. 
We use this lemma 
%together with vertex splitting and existing results
%on orthogonal drawings of 4-graphs 
to obtain small EPG-rep\-re\-sen\-ta\-tions
for a number of graph classes (we will not give formal definitions
of these graph classes; see \cite{Die12}).

\begin{corollary}
%Trees, outer-planar graphs, series-parallel graphs, and more generally all
All graphs of bounded treewidth 
(in particular, trees, outer-planar graphs and series-parallel graphs)
have an EPG-rep\-re\-sen\-ta\-tion
%in an $O(\sqrt{n})\times O(\sqrt{n})$-grid, thus 
in $O(n)$ area. Graphs of bounded genus have an EPG-rep\-re\-sen\-ta\-tion 
in % an $O(\sqrt{n}\log n)\times O(\sqrt{n}\log n)$-grid, thus in
$O(n\log^2n)$ area.
\end{corollary}
\begin{proof}
Let $G$ be one such graph for which we wish to obtain the
EPG-rep\-re\-sen\-ta\-tion.
%,and assume that $G$ has treewidth $k$ respectively is planar.
$G$ may not be a 4-graph, but we can turn it into a 4-graph by {\em 
vertex-splitting}, defined as follows.  Let $v$ be a vertex with 5 or
more neighbours $w_1,\dots,w_d$.   Create a new vertex $v'$, which is
adjacent to $w_1,w_2,w_3$ and $v$, and delete the edge $(v,w_i)$ for $i=1,2,3$.
Observe that $\deg(v')=4$ and $\deg(v)$ is reduced by 2, so sufficient
repetition ensures that all vertex degrees are at most 4.   Let $H$
be the resulting graph, and observe that $G$ is a minor of $H$.

Every vertex $v$ of $G$ gives rise to at most $\deg(v)/2$ new vertices
in $H$, so $H$ has at most $n+m$ vertices.  Since graphs of bounded
treewidth have $O(n)$ edges and graphs of bounded genus have $O(n)$ edges, therefore
$H$ has $O(n)$ vertices. Markov and Shi~\cite{markov.shi:constant} argued
that the splitting can be done in such a way that $tw(H)\leq tw(G)+1$.
%, we can do vertex splitting on graphs of treewidth
%$k$ such that the result has treewidth at most $k+1$. Thus one can
%ensure that $H$ also has bounded treewidth. 
It is also not hard to see
that with a suitable way of splitting, one can ensure that in the case of
bounded genus graphs the graph $H$ obtained by splitting has the same genus. 

By Leiserson's construction \cite{Lei80}, 4-graphs of bounded treewidth have
an orthogonal drawing in $O(n)$ area and those of bounded genus have an orthogonal
drawing in $O(n\log^2 n)$ area.  
%Closer inspection of the construction yield
%that the grid-size is $O(\sqrt{n})\times O(\sqrt{n})$ and
%$O(\sqrt{n}\log n)\times O(\sqrt{n}\log n)$, respectively.  
%Lemma~\ref{lem:orthTransform} now gives the result. \qed
\qed
\end{proof}

For classes of 4-graphs, Lemma~\ref{lem:orthTransform} and
Leiserson's construction \cite{Lei80} give directly the following
stronger results:

\begin{corollary}\label{cor:4graphsEPG}
Hereditary classes of $4$-graphs that have balanced separators of size
$O(n^\epsilon)$ with $\epsilon <1/2$ have EPG-rep\-re\-sen\-ta\-tion
in $O(n)$ area. Hereditary classes of $4$-graphs that have balanced separators of size $O(n^\epsilon)$ with $\epsilon >1/2$ have EPG-rep\-re\-sen\-ta\-tion
in $O(n^{2\epsilon})$ area.
\end{corollary}

The first bound in Corollary~\ref{cor:4graphsEPG} is tight thanks
to Theorem~\ref{thm:numedg}. The second bound is tight thanks
to Corollary~\ref{cor:pathwidth} and the fact that there are such
classes of graphs which contain triangle-free graphs of pathwidth
$\Omega(n^\epsilon)$, for example the class of finite 4-graphs that are subgraphs of the 3D
integer grid with $\epsilon=2/3$.

\iffull
%%%%%%%%%%%%%%%%%%%%%%%%%%%%%%%%%%%%%%%%%%%%%%%%%%%%%%%%%%%%%%%%%%%%%%%%
\section{Conclusion}\label{sec:conclusion}

In this paper, we study the problem of creating EPG-rep\-re\-sen\-ta\-tions
of graphs, with the objective of small area of the supporting
grid.  A very simple construction shows that an $O(n)\times O(n)$-grid
is feasible for all graphs.  We show that a smaller grid can be achieved
only if the pathwidth is smaller, presuming at most two vertices may use
a grid-edge. Vice versa, graphs of small pathwidth
can be drawn in a small grid.  Via a detour into orthogonal drawings, we
also argue that graphs of bounded treewidth and planar graphs have
EPG-rep\-re\-sen\-ta\-tions where the grid-size is $O(\sqrt{n})\times O(\sqrt{n})$
and $O(\sqrt{n}\log n)\times O(\sqrt{n}\log n)$, respectively.

As for open problems, the main question is for what other graph classes
we can create EPG-rep\-re\-sen\-ta\-tions of smaller area.  If we restrict
our attention to EPG-rep\-re\-sen\-ta\-tions where at most two edges may use
a grid-edge, then the pathwidth seems the be the right measure for
the height of such representations.  But if many vertex-paths are
allowed to share a grid-edge, then sometimes EPG-rep\-re\-sen\-ta\-tions of
much smaller height are possible, e.g. for the complete graph.  We
suspect that the {\em clique-cover-number}, i.e., the minimum number
of cliques needed to cover all edges, may be a relevant parameter here.
\fi

%%%%%%%%%%%%%%%%%%%%%%%%%%%%%%%%%%%%%%%%%%%%%%%%%%%%%%%%%%%%%%%%%%%%%%%%
\newpage
\bibliographystyle{splncs03}
\bibliography{bibs}

%%%%%%%%%%%%%%%%%%%%%%%%%%%%%%%%%%%%%%%%%%%%%%%%%%%%%%%%%%%%%%%%%%%%%%%%
%\begin{appendix}

%\section{Appendix Stuff}

%\end{appendix}

\end{document}